\documentclass[a4paper,autoref]{lipics-v2021}
\hideLIPIcs
\nolinenumbers
\newif\ifISAAC
\ISAACtrue

\EventEditors{John Q. Open and Joan R. Access}
\EventNoEds{2}
\EventLongTitle{42nd Conference on Very Important Topics (CVIT 2016)}
\EventShortTitle{CVIT 2016}
\EventAcronym{CVIT}
\EventYear{2016}
\EventDate{December 24--27, 2016}
\EventLocation{Little Whinging, United Kingdom}
\EventLogo{}
\SeriesVolume{42}
\ArticleNo{23}
\author{Mark de Berg}{Department of Mathematics and Computer Science, TU Eindhoven, the Netherlands}{M.T.d.Berg@tue.nl}{https://orcid.org/0000-0001-5770-3784}{}
\author{Bart M.P.~Jansen}{Department of Mathematics and Computer Science, TU Eindhoven, the Netherlands}{B.M.P.Jansen@tue.nl}{https://orcid.org/0000-0001-8204-1268}{}
\author{Jeroen S.K.~Lamme}{Department of Mathematics and Computer Science, TU Eindhoven, the Netherlands}{J.S.K.Lamme@tue.nl}{https://orcid.org/0009-0005-8901-2271}{}
\funding{MdB and BJ are supported by the  Dutch Research Council (NWO) through Gravitation-grant NETWORKS-024.002.003.}
\ccsdesc[100]{Theory of computation~Design and analysis of algorithms}
\title{Star-Based Separators for Intersection Graphs of $c$-Colored Pseudo-Segments}
\authorrunning{M.~de Berg, B.M.P.~Jansen, and J.S.K.~Lamme}
\Copyright{Mark de Berg and Bart M.P.~Jansen and Jeroen Lamme}
\keywords{Computational geometry, intersection graphs, biclique-based separators, distance oracles}

\usepackage{amsmath,amssymb,amsfonts,latexsym}
\usepackage{color,enumerate,graphicx}
\usepackage{url,hyperref}
\usepackage{fancybox}
\usepackage{algorithm}
\usepackage[noend]{algpseudocode}
\usepackage{mathtools}
\usepackage{xspace}
\usepackage{thmtools} 
\usepackage{thm-restate}
\definecolor{light-gray}{gray}{0.95}
\floatstyle{plaintop}
\newcommand{\eps}{\varepsilon}
\newcommand{\etal}{\emph{et al.}\xspace}

\theoremstyle{plain}
\newenvironment{myquote}%
  {\list{}{\leftmargin=4mm\rightmargin=4mm}\item[]}%
  {\endlist}
\newenvironment{claiminproof}{\begin{myquote}\noindent\emph{Claim.}}{\end{myquote}}
\newenvironment{proofinproof}{\begin{myquote}\noindent\emph{Proof.}}{\hfill $\lhd$ \end{myquote}}
\newcommand{\C}{\ensuremath{\mathcal{C}}}

\newcommand{\G}{\ensuremath{\mathcal{G}}}
\newcommand{\cH}{\ensuremath{\mathcal{H}}}
\newcommand{\cP}{\ensuremath{\mathcal{P}}}

\newcommand{\T}{\ensuremath{\mathcal{T}}}

\newcommand{\REAL}{\ensuremath{\mathbb{R}}}
\newcommand{\Reals}{\REAL}

\renewcommand{\leq}{\leqslant}
\renewcommand{\geq}{\geqslant}
\renewcommand{\le}{\leqslant}

\newcommand{\bd}{\partial}

\newcommand{\graph}{\G}
\newcommand{\ig}{\graph^{\times}}
\newcommand{\sep}{S}

\newcommand{\actF}{F}
\newcommand{\inactF}{\overline{\actF}}
\newcommand{\mystar}{\mathrm{star}}

\newcommand{\mis}{{\sc Independent Set}\xspace}
\newcommand{\domset}{{\sc Dominating Set}\xspace}
\newcommand{\qcol}{$q$-{\sc Coloring}\xspace}
\newcommand{\fvs}{{\sc Feedback Vertex Set}\xspace}
\newcommand{\seg}{\mathrm{seg}}

\newcommand{\hd}{\mbox{\sc hd}}
\newcommand{\vd}{\mbox{\sc vd}}

\begin{document}
\maketitle
%------------------------------------------------------------------------------------------

%------------------------------------------------------------------------------------------
\begin{abstract}
The Planar Separator Theorem, which states that any planar graph $\graph$  has a separator 
consisting of~$O(\sqrt{n})$ nodes whose removal partitions $\graph$ into components of size
at most $\tfrac{2n}{3}$, is a widely used tool to obtain fast algorithms on planar graphs. 
Intersection graphs of disks, which generalize planar graphs, do not admit such separators. 
It has recently been shown that disk graphs \emph{do} admit so-called clique-based separators 
that consist of $O(\sqrt{n})$ cliques. This result has been generalized to intersection graphs
of various other types of disk-like objects. Unfortunately, segment intersection graphs 
do not admit small clique-based separators, because they can contain arbitrarily large bicliques.
This is true even in the simple case of axis-aligned segments.

In this paper we therefore introduce \emph{biclique-based separators}
(and, in particular, \emph{star-based separators}), which are separators consisting of a small
number of bicliques (or stars). We prove that any $c$-oriented set of $n$ segments 
in the plane, where $c$ is a constant, admits a star-based separator consisting of $O(\sqrt{n})$ stars. 
In fact, our result is more general, as it applies to any set of~$n$ pseudo-segments that
is partitioned into $c$ subsets such that the pseudo-segments in the same subset are pairwise disjoint.
We extend our result to intersection graphs of $c$-oriented polygons.
These results immediately lead to an almost-exact distance oracle for such intersection graphs,
which has~$O(n\sqrt{n})$ storage and~$O(\sqrt{n})$ query time, and that can report
the hop-distance between any two query nodes in the intersection graph with an additive error of at most~2.
This is the first distance oracle for such types of intersection graphs that has
subquadratic storage and sublinear query time and that only has an additive error.
\end{abstract}
%------------------------------------------------------------------------------------------

%------------------------------------------------------------------------------------------
\section{Introduction}
\label{sec:intro}
%------------------------------------------------------------------------------------------

%------------------------------------------------------------------------------------------
\subparagraph{Background.}
%------------------------------------------------------------------------------------------
The celebrated Planar Separator Theorem by Lipton and Tarjan~\cite{planarseparator} states 
that any planar graph admits a balanced separator of size $O(\sqrt{n})$. More precisely, 
for any planar graph~$G=(V,E)$ with $n$ nodes there exists a set~$S\subseteq V$ of size $O(\sqrt{n})$ 
whose removal partitions $G$ into components of at most $\tfrac{2n}{3}$ nodes each. 
This fundamental tool has been used to develop efficient algorithms for many classic problems on planar graphs. 

Geometric intersection graphs---graphs whose nodes correspond to objects in the plane and 
that have an edge between two nodes iff the corresponding objects intersect---are a generalization 
of planar graphs that have received widespread attention, in computational geometry, graph theory, and
parametrized complexity. (For an overview of work in the latter area, we refer the reader to the survey by Xue~and~Zehavi~\cite{XUE2025100796}.)
Unfortunately, geometric intersection graphs 
do not admit balanced separators of sublinear size, because they can contain arbitrarily large cliques.
This led De~Berg~\etal~\cite{bbkmz-ethf-20} to introduce so-called \emph{clique-based separators}:
balanced separators that consist of a small number of disjoint (but potentially large) cliques instead of a small number of nodes.
They proved that any disk graph---and more generally, any intersection graph of convex fat 
objects in the plane---admits a clique-based separator of size $O(\sqrt{n})$. Here the
\emph{size} of a clique-based separator is the number of cliques it consists of.
They also showed how to use such clique-based separators to obtain sub-exponential algorithms 
for various classic graph problems, including \mis, \domset, and \fvs.
Recently, De Berg~\etal~\cite{bkmt-cbsgis-23} showed that 
intersection graphs of pseudo-disks, and intersection graphs
of geodesic disks inside a simple polygon, admit balanced separators 
consisting of $O(n^{2/3})$ cliques, and Aronov~\etal~\cite{abt-cbsgd-24}
proved that intersection graphs of geodesic disks in any well-behaved 
metric in the plane admit balanced separators consisting of $O(n^{3/4+\eps})$ cliques.
\medskip

One may wonder if \emph{all} geometric intersection graphs have sublinear clique-based separators.
Unfortunately the answer is no, even for intersection graphs of 
horizontal and vertical line segments. The problem is that such graphs
can contain arbitrarily large bicliques, and $K_{n,n}$ does not admit a
sublinear clique-based separator. We therefore introduce 
\emph{biclique-based separators}, which are separators consisting of bicliques,
and we show that any set of horizontal and vertical segments admits a balanced 
separator consisting of a small number of bicliques. In fact, our result
is stronger (as it uses \emph{star graphs} in the separator, and not just any biclique)
and it applies to a much wider class of intersection graphs, as discussed next.

%------------------------------------------------------------------------------------------
\subparagraph{Our contribution.}
%------------------------------------------------------------------------------------------
Let $\G=(V,E)$ be an undirected graph with $n$ nodes. 
A collection~$\sep=\{\sep_1,\ldots\sep_t\}$ of (not necessarily induced) disjoint subgraphs from $G$ is called a 
balanced\footnote{In the sequel we will often omit the adjective
\emph{balanced} and simply speak of \emph{separators}.}
\emph{biclique-based separator} for $G$ if it has the following properties:
\begin{itemize}
\item Each subgraph $\sep_i$ is a biclique.
\item The removal from $\G$ of all subgraphs $\sep_i$ and their incident edges partitions $\G$
      into connected components with at most $\frac{2n}{3}$ nodes each.
\end{itemize}
The \emph{size} of a biclique-based separator is the number of bicliques it is comprised of.
If each biclique~$\sep_i\in\sep$ is a star, then we call $\sep$ a \emph{star-based separator}.
Note that the subgraphs $\sep_i$ in a biclique-based separator (or: in a star-based separator)
need not be induced subgraphs of~$\G$. This is necessary to be able to handle large cliques.
To avoid confusion between our new separators that are comprised of bicliques and the traditional 
separators that are comprised of individual nodes, we will refer to the latter as \emph{node-based separators}. 
\medskip

We denote the intersection graph induced by a set $V$ of $n$ objects in the plane by~$\ig[V]$.
Thus, the nodes in $\ig[V]$ are in one-to-one correspondence with the objects in~$V$
and there is an edge between two objects $u,v\in V$ iff $u$ intersects~$v$.
In Section~\ref{sec:separator} we prove that a star-based separator of size $O(\sqrt{n})$ 
exists for the intersection graph of any set~$V$ of axis-parallel segments;
the bound on the size of the separator is tight in the worst case.
In fact, we will prove that a star-based separator of size $O(\sqrt{n})$ exists
for any set $V$ of pseudo-segments\footnote{A set $V$ of curves in the plane is a 
set of pseudo-segments if any two curves in $V$ are either disjoint or intersect in a single point 
that is a proper crossing (and not a tangency).} 
that is partitioned into subsets $V_1,\ldots,V_c$ such that the pseudo-segments
from each $V_i$ are disjoint from each other.
In other words, each $V_i$ is an independent set in $\ig[V]$.
We call such a set $V$ a \emph{$c$-colored set} of pseudo-segments,
and we call $V_1,\ldots,V_c$ its color classes; see Figure~\ref{fig:overview}(left). 
Note that a set of axis-parallel segments such that no two segments of the same orientation
intersect, is a $2$-colored set of pseudo-segments. More generally, a~$c$-oriented set of line 
segments such that no two segments from the same orientation intersect, is a
$c$-colored set of pseudo-segments. Intersection graphs of $c$-oriented segments
are often referred to as {\sc $c$-dir} graphs, and when segments of the same
orientation are not allowed to intersect, they are referred to as {\sc pure $c$-dir} graphs~\cite{KRATOCHVIL1994233,Kratochvíl1990}.  In Section \ref{sec:algorithms} we show that a star-based separator of size $O(\sqrt{n})$ for a $c$-{\sc dir} graph can be computed in $O(n \log n)$ time, if the segments which induce this graph are given. 

In Section~\ref{sec:polygons-and-strings} we extend our result to the case where $V$ 
is a set of $n$ constant-complexity~$c$-oriented
polygons, that is, a set of polygons such that the set of
edges of all polygons is a~$c$-oriented set. Note that two polygons
can intersect without having their boundaries intersect, namely when one
polygon is completely contained in the outer boundary of the other polygon---this is the main
difficulty we need to handle when extending our results to polygons. 
In Section \ref{sec:algorithms} we show that this separator can also be computed in $O(n \log n)$ time.
Finally, in Section~\ref{sec:polygons-and-strings} we also present a straightforward
greedy algorithm that computes a star-based separator of size $O(n^{2/3} \log ^{2/3} n)$ 
for any string graph. 

%------------------------------------------------------------------------------------------
\medskip
\noindent\emph{Application to distance oracles.}
%------------------------------------------------------------------------------------------
A \emph{distance oracle} for a (potentially weighted) graph $\graph=(V,E)$ is a data structure that can
quickly report the distance between two query nodes~$s,t\in V$. Such queries
can trivially be answered in $O(1)$ time if we store the distance between
any two nodes in a distance table, but this requires $\Omega(n^2)$ storage. 
The challenge is to design distance oracles that use subquadratic storage. 
Unfortunately, this is not possible in general: any distance oracle must use 
$\Omega(n^2)$ bits of storage in the worst case, irrespective of the query 
time~\cite{DBLP:journals/jacm/ThorupZ05}. This is even true for distance
oracles that approximate distances to within a factor strictly less than~$3$.
Thus, work on distance oracles  concentrated on special graph classes and,
in particular, on planar graphs. More than two decades of research culminated in an 
exact distance oracle for weighted planar graphs that uses $O(n^{1+o(1)})$ storage 
and has~$O(\log^2n)$ query time~\cite{DBLP:journals/jacm/CharalampopoulosGLMPWW23}. 
For the unweighted case---in other words, if we are interested in the \emph{hop-distance}---there 
is a $(1+\eps)$-approximate distance oracle with $O(1/\eps^2)$ query time and~$O(n/\eps^2)$ storage~\cite{DBLP:conf/focs/LeW21}. See the survey by Sommer~\cite{DBLP:journals/csur/Sommer14}
and the paper by Charalampopoulos~\etal~\cite{DBLP:journals/jacm/CharalampopoulosGLMPWW23} for overviews
of the existing distance oracles for various graph classes.

For geometric intersection graphs, only few results are known. 
Gao and Zhang~\cite{gao-zhang}, and Chan and Skrepetos~\cite{DBLP:journals/jocg/ChanS19a},
provide $(1+\eps)$-approximate distance oracles with $O(n\log n)$ storage
and $O(1)$ query time for weighted unit-disk graphs.
No exact distance oracles %, or distance oracles with only an additive error,
that use subquadratic storage and have sublinear
query time are known, even for unweighted unit-disk graphs.
Very recently, Chang, Gao, and Le~\cite{DBLP:conf/compgeom/Chang0024}
presented an \emph{almost exact} distance-oracle for unit-disk graphs
(and, more generally, for intersection graphs of similarly sized,
convex, fat pseudodisks) that uses $O(n^{2-1/18})$ storage
and that can report the distance between two query nodes,
up to an additive error\footnote{In the most recent arxiv version~\cite{DBLP:journals/corr/abs-2401-12881}
the error has been reduced to 1.} of~2, in $O(1)$ time.
Another recent result is by Aronov, De~Berg, and Theocharous~\cite{abt-cbsgd-24},
who presented an almost exact distance oracle that uses clique-based separators.
For intersection graphs of geodesic disks in the plane,
their oracle uses $O(n^{7/4+\eps})$ storage, has $O(n^{3/4+\eps})$ query time,
and can report the hop-distance between two query points up to an 
additive error of~$1$. For Euclidean disks, the storage and preprocessing would be $O(n\sqrt{n})$
and $O(\sqrt{n})$, respectively. As we explain later,
their approach also works with biclique-based separators; the only
difference is that the additive error increases from~$1$ to~$2$.
Thus, we obtain an almost exact distance oracle for intersection graphs of
$c$-colored pseudo-segments with $O(n\sqrt{n})$ storage and 
$O(\sqrt{n})$ query time. This is, to the best of our knowledge, the first
almost exact distance oracle for intersection graphs of non-disk-like objects.

%------------------------------------------------------------------------------------------
\section{A star-based separator for $c$-colored pseudo-segments}
\label{sec:separator}
%------------------------------------------------------------------------------------------
Let $V$ be a $c$-colored set of pseudo-segments, as defined above.
To simplify the terminology, from now on we simply refer to the pseudo-segments 
in $V$ as \emph{segments}. 
We assume that the segments in $V$ are in general position and,
in particular, that no three segments meet in a common point 
and that no endpoint of one segment lies on another segment. 
This assumption is without loss of generality,
as it can always be ensured by perturbing the segments slightly.

%------------------------------------------------------------------------------------------
\subparagraph{The construction.}
%\label{subsec:construction}
%------------------------------------------------------------------------------------------
Recall that a \emph{contact graph} is the intersection graph of a set of interior-disjoint objects. 
Contact graphs of curves are known to be planar if no four objects meet in a common point~\cite[Lemma 2.1]{HLINENY199859}.
Our strategy to construct a star-based separator~$\sep$ for~$\ig[V]$ consists of the following steps, illustrated in Figure~\ref{fig:overview} and explained in more detail later.
%---------------------------------------------------------------------------------------
\begin{figure}[t]
\begin{center}
\includegraphics{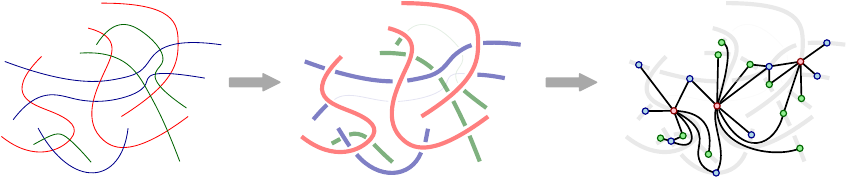}
\end{center}
\caption{Left: A $3$-colored set of pseudo-segments. Middle: The active fragments created by our algorithm. Right: The contact graph induced by the active fragments. }
\label{fig:overview}
\end{figure}
%---------------------------------------------------------------------------------------
\smallskip
\begin{description}
\item[\rm\emph{Step~1.}] We partition each segment in~$V$ into \emph{fragments}. Some fragments will
      be \emph{active}, while others will be \emph{inactive}. This partition will be such that 
      active segments do not cross each other, although they may touch.
\item[\rm\emph{Step~2.}] Let $\cH$ be the contact graph on the active segments. 
      We construct a separator~$\sep_{\cH}$ on~$\cH$,
      using a suitable weighting scheme on the nodes of~$\cH$.
      Because $\cH$ is planar, constructing the separator can be done using the Planar Separator Theorem.
\item[\rm\emph{Step~3.}] We use $\sep_{\cH}$ to construct our star-based separator~$\sep$ for $\ig[V]$. For a fragment $f$, let~$\seg(f)\in V$ denote the segment containing~$f$.
      Intuitively, we want to put a star into $\sep$ for each fragment~$f$ in the separator
      $\sep_{\cH}$, namely, the star consisting of the segment~$\seg(f)$ as well as all other segments intersecting~$\seg(f)$. For technical reasons, however,
      we actually have to put a slightly larger collection of stars into~$\sep$.
\end{description}
\smallskip
To make this strategy work, we need to control the size of the contact graph~$\cH$.
More precisely, to obtain a star-based separator of size $O(\sqrt{n})$,
the size of $\cH$ needs to be~$O(n)$. Thus we cannot, for example,
cut each segment into fragments at its intersection points with all other segments and make all
the resulting fragments active. On the other hand, if we ignore certain parts of the segments
by making them inactive, we miss certain intersections and we run the risk that our final set 
of stars is no longer a valid separator. Next, we describe
how to overcome these problems by carefully creating the active fragments.

%------------------------------------------------------------------------------------------
\subparagraph{\em Step~1: Creating the active fragments.}
%------------------------------------------------------------------------------------------
To construct the active fragments that form the nodes in our contact graph~$\cH$,
we will go over the subsets~$V_1,\ldots,V_c$ one by one. We denote the active and
inactive fragments created for a subset~$V_i$ by $\actF_i$ and $\inactF_i$, respectively.
For~$1\leq i\leq c$, we define~$\actF_{\leq i} := \actF_1 \cup\cdots\cup \actF_i$, and 
we define $\actF_{<i}$ similarly. 
\medskip

Handling the first subset~$V_1$ is easy: we simply define $\actF_1 := V_1$. In other words,
each segment in $V_1$ becomes a single fragment, and all these fragments are active.

Now consider a subset~$V_i$ with~$i>1$. Each segment $v\in V_i$ is partitioned into
one or more fragments by cutting it at every intersection point of~$v$ with an active
fragment~$f\in \actF_{<i}$. Let~$X_i$ be the set of fragments thus created. 
There are two types of fragments in~$X_i$: fragments~$f$ that contain 
an endpoint of the segment $v\in V_i$ contributing~$f$---there are at least one and at most two
of these fragments per segment~$v\in V_i$---and fragments that do not contain such an endpoint. We call fragments of the former type \emph{end fragments} and
fragments of the latter type \emph{internal fragments}.
Note that an internal fragment has its endpoints on two distinct active fragments 
$g,g'\in\actF_{<i}$. We then say that $f$ \emph{connects} $g$ and~$g'$.

Now that we have defined $X_i$, we need
to decide which fragments in~$X_i$ become active. To avoid making too many
fragments active, we will partition~$X_i$ into equivalence classes, and we will activate
only one fragment from each equivalence class. To define the equivalence classes
we first define, for two internal fragments $f,f'\in X_i$ that connect the same 
pair of fragments $g,g'\in \actF_{<i}$, a region $Q(f,f',g,g')$, as follows;
see Figure~\ref{fig:def-Q}(i) for an illustration. 

First, suppose that the segments $g$ and $g'$ do not touch each other, as in the left part of Figure~\ref{fig:def-Q}(i).
Thus, $\Reals^2 \setminus (g \cup g')$ is a single, unbounded region with two ($1$-dimensional) holes, namely $g$ and $g'$. 
The fragment $f$ connects these two holes, and so $\Reals^2 \setminus (g \cup g' \cup f)$ 
is still a single unbounded region, but now with one hole. Removing $f'$ 
from this region splits it 
into two regions, one bounded and one unbounded. We define $Q(f,f',g,g')$ to be the bounded region.
Now suppose that $g$ and $g'$ touch each other, say at an endpoint of $g$. We slightly shrink $g$ 
at the point where it touches $g'$, and then define $Q(f,f',g,g')$ as above. Note that in this case 
$Q(f,f',g,g')$ may consist of one or two bounded regions, if we undo the shrinking process; 
see the middle and right part in Figure~\ref{fig:def-Q}(i). 
We can now define the equivalence classes.
%------------------------------------------------------------------------------------------
\begin{figure}
\begin{center}
\includegraphics{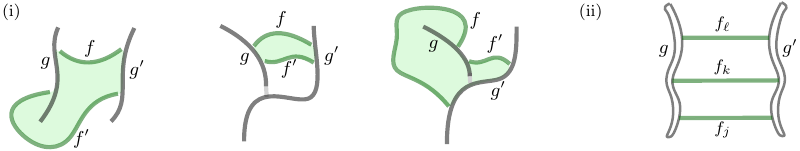}
\end{center}
\caption{(i) Defining the region $Q(f,f',g,g')$.  (ii) Illustration for the proof of Lemma \ref{lem:equivalence}.} 
\label{fig:def-Q}
\end{figure}
%------------------------------------------------------------------------------------------
\begin{definition}\label{def:equivalence}
Let $f$ and $f'$ be two fragments in~$X_i$. We say that $f$ and $f'$ are \emph{equivalent},
denoted by $f\equiv f'$, if $f=f'$ or the following two conditions hold.
\begin{enumerate}[(i)]
\item The fragments $f$ and $f'$ are internal and connect the same pair of fragments~$g,g' \in\actF_{<i}$.
\item The region~$Q(f,f',g,g')$ enclosed by the fragments $f,f',g,g'$ does not contain
      an endpoint of any segment in~$V$.  
\end{enumerate}
\end{definition}
%------------------------------------------------------------------------------------------
%------------------------------------------------------------------------------------------

The following lemma shows that~$\equiv$ is indeed an equivalence relation.

\begin{lemma}\label{lem:equivalence}
    The relation~$\equiv$ defined in Definition \ref{def:equivalence} is an equivalence relation.
\end{lemma}
\begin{proof}
It is clear that $\equiv$ is reflexive and symmetric. It remains to show that if~$f_1\equiv f_2$
and~$f_2\equiv f_3$, then~$f_1\equiv f_3$. Let $g,g'$ be the fragments connected by~$f_1,f_2,f_3$.
We can assume that $g$ and $g'$ do not touch; otherwise, as before, we can
shrink one of the fragments so that the arguments below still apply.
It is instructive to view $g$ and $g'$ as being slightly inflated so that they become closed
curves and no longer have two ``sides''. We then see that, from a topological point of view,
the situation is always as in Figure~\ref{fig:def-Q}(ii): two of the fragments,~$f_j$ and~$f_\ell$,
are incident to the unbounded face, while the third fragment~$f_k$ is not. Thus, if we define
$Q_{jk} := Q(f_j,f_k,g,g')$ and we define $Q_{k\ell}$ and $Q_{j\ell}$ similarly, then~$Q_{j\ell} = Q_{jk}\cup Q_{k\ell}$. This implies that, no matter which of the three fragments
$f_j,f_k,f_\ell$ is $f_2$, we always have~$Q_{13} \subseteq Q_{12}\cup Q_{23}$.
Since $Q_{12}$ and $Q_{23}$ do not contain endpoints of segments in $V$ if  $f_1\equiv f_2$
and $f_2\equiv f_3$, neither does~$Q_{13}$. Thus, $f_1\equiv f_3$.
\end{proof}

%------------------------------------------------------------------------------------------

%------------------------------------------------------------------------------------------
We now partition $X_i$ into equivalence classes according to the relation~$\equiv$ defined above.
For each equivalence class, we make an arbitrary fragment $f\in X_i$ from that class active
and put it into~$\actF_i$; the other fragments from that equivalence class are
made inactive and put into~$\inactF_i$. Note that end fragments 
are always active, since they do not connect a pair of 
fragments from~$\actF_{<i}$ and thus cannot be equivalent to any other fragment.

%------------------------------------------------------------------------------------------
\subparagraph{\em Step~2: Creating the contact graph~$\cH$ and its separator~$\sep_{\cH}$.}
%------------------------------------------------------------------------------------------
Let $\actF := \actF_1 \cup\cdots \cup \actF_c$ be the set of active fragments created in Step~1,
and let $\inactF := \inactF_1 \cup\cdots \cup \inactF_c$ be the inactive fragments.
We define $\cH=(\actF,E_{\cH})$ to be the contact graph of~$\actF$. More precisely, for
two fragments~$f,f\in \actF$ we put the edge $(f,f')$ in $E_{\cH}$ iff $f$ and $f'$ are 
in contact---that is,~$f\cap f'\neq\emptyset$---and they do not belong\footnote{Since we 
do not put an edge between fragments belonging to the same segment, even if
these fragments touch, $\cH$ is formally speaking not a contact graph, but we permit ourselves this abuse of terminology.} to the same
segment in~$V$; see Figure~\ref{fig:overview}. Because of our general position assumption,
no three fragments from different segments meet in a point, and therefore~$\cH$ is planar~\cite[Lemma 2.1]{HLINENY199859}.
From now on, with a slight abuse of notation, we will not make a distinction between
the nodes in $\cH$ and the corresponding fragments in~$\actF$.

We now wish to create a separator $\sep_{\cH}$ for $\cH$. In Step~3 we will use $\sep_{\cH}$
to create a separator $\sep$ for $\ig[V]$. To ensure that $\sep$ will be balanced,
we will put weights on the nodes in $\cH$ and use a weighted version of the
Planar Separator Theorem, as described next.

For each segment $s\in V$, we designate one of its end fragments---recall that
end fragments are always active---as its \emph{representative fragment}. 
We give all representative fragments a weight of $\tfrac{1}{n}$, and all other 
fragments a weight of~$0$. Note that the total weight of the fragments in $\actF$ is~$1$. 
We apply the weighted separator theorem given below to $\cH$. This gives us a 
separator~$\sep_{\cH}$, and parts $A_{\cH}, B_{\cH} \subseteq \actF\setminus \sep{_\cH}$
such that there are no edges between $A_{\cH}$ and~$B_{\cH}$.
%------------------------------------------------------------------------------------------
\begin{lemma}[Theorem 4 in \cite{planarseparator}]\label{lem: weighted seperator}
     Let $G$ be a non-negatively weighted planar graph containing~$n$ nodes whose weights sum up to at most~$1$. The node set $V_G$ can be partitioned into a separator~$\sep$, and sets $A$ and $B$ such that $|\sep| = O(\sqrt{n})$, no edge connects $A$ and $B$, and the total weight of the nodes in $A$, as well as the total weight of the nodes in $B$, is at most $\tfrac{2}{3}$.
\end{lemma}
%------------------------------------------------------------------------------------------

%------------------------------------------------------------------------------------------
\subparagraph{\em Step~3: Creating the star-based separator $\sep$ for~$\ig[V]$.}
%------------------------------------------------------------------------------------------
Using the separator~$\sep_{\cH}$ created in Step~2, we now create our star-based
separator $\sep$ for the intersection graph~$\ig[V]$. We do this
by putting one or three stars into~$\sep$ for each fragment $f\in \sep_{\cH}$, as follows.
For a segment~$s\in V$, define $\mystar(s)$ to be the subgraph of $\ig[V]$ consisting
of $s$ and all its incident edges. Thus, the nodes in $\mystar(s)$ are the segment $s$ itself
plus the segments $s'\in V$ that intersect~$s$.
\begin{itemize}
\item If $f\in \sep_{\cH}$ is an end fragment then
      we put $\mystar(\seg(f))$ into~$\sep$.
\item If $f\in \sep_{\cH}$ is an internal fragment then let $g,g'\in F$
      be the pair of active fragments connected by~$f$.
      We put $\mystar(\seg(f))$, $\mystar(\seg(g))$, and $\mystar(\seg(g'))$ into~$\sep$.
\end{itemize}
Note that multiple copies of the same star can be added to $\sep$.
We remove these duplicates to ensure that all star graphs in $\sep$ have unique centers.
It can still be the case that several stars in $\sep$ contain the same node. 
To make the stars pairwise disjoint, we therefore remove non-center nodes 
until every node appears in at most one star in~$\sep$.

%------------------------------------------------------------------------------------------
\subparagraph{The analysis.}
%\label{subsec:analysis}
%------------------------------------------------------------------------------------------
We now show that the construction described above yields a balanced separator of the
required size. This requires proving two things: 
the \emph{separation property}, namely that the removal
of $\sep$ partitions $\ig[V]$ into components of size at most $\tfrac{2n}{3}$,
and the \emph{size property}, namely that $\sep$ consists of~$O(\sqrt{n})$ stars.

%------------------------------------------------------------------------------------------
\begin{figure}[t]
\begin{center}
\includegraphics{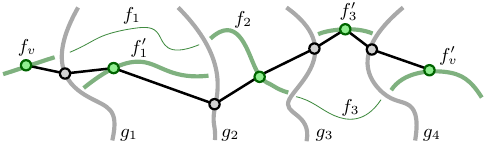}
\end{center}
\caption{The various fragments used in the proof of Lemma~\ref{lem: assignment properties of blocks}.
        Note that the segments $\seg(f'_1)$ and $\seg(f'_3)$ are not drawn in their entirety.
        The path in $\cH$ from $f_v$ to $f'_v$ is also shown.}
\label{fig:path}
\end{figure}
%------------------------------------------------------------------------------------------
%------------------------------------------------------------------------------------------
\subparagraph{\em Proving the separation property.}
%------------------------------------------------------------------------------------------
To prove the separation property, it suffices to show that~$V\setminus \sep$ can be 
partitioned\footnote{Formally, we should have written $V\setminus \bigcup \sep$
instead of $V\setminus \sep$, since $\sep_{\cH}$ is a set of stars and not a set of nodes,
but we prefer the simpler (though technically incorrect) notation. }
into subsets $A$ and $B$ such that $|A|\leq \tfrac{2n}{3}$ and $|B|\leq \tfrac{2n}{3}$, 
and such that no segment in $A$ intersects any segment in~$B$. 

We define the sets $A$ and $B$ as follows. For each segment $v \in V$ not contained in 
a star in $\sep$ we look at its representative fragment~$f_v$.
Note that $f_v$ must be contained in either $A_{\cH}$ or~$B_{\cH}$, since $f_v\in\sep_{\cH}$
would imply that $v$ is contained in a star in~$\sep$. 
If $f_v \in A_{\cH}$ then we add $v$ to $A$, else we add $v$ to~$B$. 
The value $\tfrac{|A|}{n}$ can be at most the total weight of fragments in $A_{\cH}$, and
a similar statement holds for $B$. Hence, the next observation follows from the fact 
that $\sep_{\cH}$ is a balanced separator.
%------------------------------------------------------------------------------------------
\begin{observation}\label{obs:size-of-parts}
$|A|\leq \tfrac{2n}{3}$ and $|B|\leq \tfrac{2n}{3}$.  
\end{observation}
%------------------------------------------------------------------------------------------
The more challenging part is to show that no segment in $A$ intersects any segment in~$B$.
We will need the following lemma. Recall that the segment set $V$ is partitioned into
color classes~$V_1,\ldots,V_c$, which we handled one by one to create the set $\actF$ of
active fragments.
%------------------------------------------------------------------------------------------
\begin{lemma}\label{lem: assignment properties of blocks}
Let $v \in V_i\setminus \sep$ for some $1\leq i \leq c$ and let $f$ be an active fragment.
Suppose one of the following conditions holds:
\begin{enumerate}[(i)]
\item $f \in F_{<i}$ and $v$ intersects $f$, 
\item $\seg(f)=v$, or 
\item $f$ is equivalent to an inactive fragment $f'$ such that $\seg(f')=v$.
\end{enumerate}
Then $f\in A_{\cH}$ if $v\in A$, and $f\in B_{\cH}$ if $v\in B$. 
\end{lemma}
%------------------------------------------------------------------------------------------
\begin{proof}
    We prove the lemma under the assumption that $v \in A$; the proof for $v \in B$ is analogous. 
    Let $f_v$ be the representative fragment of~$v$. Because $v \in A$, we have $f_v \in A_{\cH}$.
    We will define sets $Z_1,Z_2,Z_3$ that contain the active fragments for which 
    conditions~(i), (ii), and (iii) hold, respectively, and then argue that the 
    lemma holds for each of the three sets.
    
    Let $f'_v$ be the other end fragment of~$v$, and let
    $Z_1=\{g_1,\ldots,g_k\}$ be the ordered set of fragments from $\actF_{<i}$ that we 
    cross as we trace $v$ from $f_v$ to $f'_v$; see Figure~\ref{fig:path}. It is possible that~$f'_v$ does not exist, in this case $Z_1$ is empty.
    Note that every pair 
    $g_j,g_{j+1}\in Z_1$ is connected by an active or inactive fragment of~$v$, which we denote by $f_j$. We now define~$Z_2,Z_3$ as follows.
    \begin{itemize}
        \item $Z_2$ contains the fragments $f_j$ that are active plus
              the end fragments $f_v$ and (if it exists) $f'_v$. Thus, $Z_2$ simply contains all active fragments of~$v$.
       \item $Z_3$ contains, for each inactive fragment $f_j$, the unique equivalent 
             active fragment $f'_j\in F_i$. 
    \end{itemize}
    It is easily checked that the sets $Z_1,Z_2,Z_3$ indeed contain exactly those fragments 
    for which conditions~(i), (ii), and (iii) hold, respectively. 
    We will now prove that all fragments in~$Z_1\cup Z_2\cup Z_3$ are in $A_{\cH}$.
    To this end, observe that the fragments in $Z_1\cup Z_2\cup Z_3$ correspond to nodes
    in $\cH$ that form a path~$\pi$ starting at $f_v$ and ending at~$f'_v$,
    as illustrated in Figure~\ref{fig:path}. Recall that $f_v\in A_{\cH}$.
    Now assume for a contradiction that there is a fragment $f\in Z_1\cup Z_2\cup Z_3$ that is in~$B_{\cH}\cup \sep_{\cH}$.
    Because there are no edges between $A_{\cH}$ and $B_{\cH}$, this means
    there must be a fragment $f^*$ on the subpath of $\pi$ from $f_v$ to $f$
    that is an element of the separator~$\sep_{\cH}$. 
    If $f^*\in Z_1\cup Z_2$ then $v\in \mystar(\seg(f^*))$, which contradicts that $v\in A$.
    Otherwise $f^*\in Z_3$ and $f^*$ connects two fragments $g_j,g_{j+1}\in Z_1$. 
    Since $f^*\in\sep_{\cH}$, this implies that $\mystar(\seg(g_j))$ is a star in~$S$. 
    Because $v\in \mystar(\seg(g_j))$, this again contradicts that $v\in A$.  
\end{proof}
%------------------------------------------------------------------------------------------
We can now prove that $\ig[V]$ does not contain edges between nodes in $A$ and
nodes in~$B$.
%------------------------------------------------------------------------------------------
\begin{lemma}\label{lem:path}
No segment $a\in A$ intersects any segment $b\in B$. 
\end{lemma}
%------------------------------------------------------------------------------------------
\begin{proof}
Assume for a contradiction that $a \in A$ and $b \in B$ intersect. Let $f_a\subseteq a$ and $f_b\subseteq b$
be the fragments containing the intersection point. We distinguish three cases.
\begin{itemize}
\item \emph{Both $f_a$ and $f_b$ are active.} Then $f_a$ and $f_b$ satisfy condition~(ii)
    of Lemma \ref{lem: assignment properties of blocks}, and so~$f_a \in A_{\cH}$ and $f_{b} \in B_{\cH}$. 
    Because $f_a$ and $f_b$ are active and intersect, $(f_a,f_b)$ is an edge in~$\cH$. 
    But then $\sep_{\cH}$ would not be a proper separator, and we reach a contradiction.
\item \emph{Both $f_a$ and $f_b$ are inactive.} Let $f_a \in \inactF_i$ and $f_b \in \inactF_j$, 
    and assume without loss of generality that $i \leq j$. Because $a$ and $b$ intersect
    and segments from the same color class are disjoint, we cannot have $i=j$. Hence, $i < j$.
    Because $f_a$ is inactive, it must be an internal fragment that connects 
    some fragments $g,g' \in \actF_{<i} \subseteq \actF_{<j}$. 
    Thus, $g,g'$ satisfy condition~(i) of Lemma \ref{lem: assignment properties of blocks},
    with $a$ playing the role of~$v$, which implies that $g,g' \in A_{\cH}$. 
    Let~$f_a'\in \actF_i$ be the active fragment that is equivalent to $f_a$. 
    Then $f'_a$ satisfies condition~(iii) from  Lemma \ref{lem: assignment properties of blocks} 
    and so $f_a' \in A_{\cH}$. The fragments $f_a,f_a',g,g'$ enclose some region $Q$. 
    The segment $b$ intersects $f_a$ so it is partly contained in $Q$;
    see Figure~\ref{fig:enteringQ} for an illustration of the possibilities for $f_b$ entering $Q$.
    %------------------------------------------------------------------------------------------
    \begin{figure}[t]
    \begin{center}
    \includegraphics{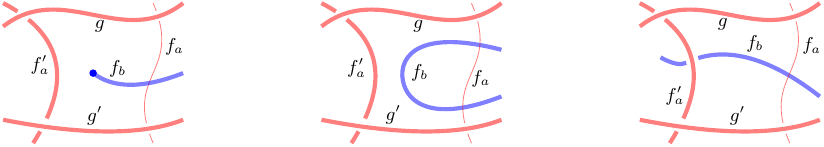}
    \end{center}
    \caption{Left: The segment $b$ has an endpoint in $Q$. Middle: The segment $b$ intersects $a$ twice. Right: The segment $b$ intersects $f_a'$. All of these cases lead to a contradiction. }
    \label{fig:enteringQ}
    \end{figure}
    %------------------------------------------------------------------------------------------
    It cannot have an endpoint in $Q$, because~$f_a$ and~$f_a'$ are equivalent. 
    Segment $b$ cannot intersect $f_a$ twice, because $V$ is a collection of pseudo-segments. 
    It follows that $b$ must intersect $f_a$, $g$ or $g'$. Now observe that
    $f_a$, $g$ and $g'$ are all fragments in~$\actF_{<j}$. This implies
    that the fragment that intersects~$b$ satisfies condition~(i) from 
    Lemma \ref{lem: assignment properties of blocks}, with $b$ playing the role of~$v$.
    But then the intersected fragment would be in $B_{\cH}$,
    contradicting that $A_{\cH}$ and $B_{\cH}$ are disjoint. 
\item  \emph{One of the fragments $f_a,f_b$ is active and one is inactive.}
    Assume wlog that $f_a$ is inactive and $f_b$ is active. 
    Let $f_a \in \inactF_i$ and $f_b \in \actF_j$. Because $a$ and $b$ intersect,
    we have $i \neq j$. 
    If~$i<j$ then the arguments from the previous case can be used
    to obtain a contradiction---indeed, these arguments did not use that $f_b$ is inactive.
    If $i > j$ then $f_b \in F_{<i}$. Also note that $a$ intersects $f_b$. 
    But then $f_b$ satisfies condition~(i) from Lemma \ref{lem: assignment properties of blocks},
    with $a$ playing the role of~$v$. It follows that $f_b \in A_{\cH}$,
    which contradicts that $A_{\cH}$ and $B_{\cH}$ are disjoint.
\end{itemize}
We have reached a contradiction in each case, thus proving the lemma.
\end{proof}
%------------------------------------------------------------------------------------------

%------------------------------------------------------------------------------------------
\subparagraph{\em Proving the size property.}
%------------------------------------------------------------------------------------------
Because we add at most three stars to $\sep$ per fragment in $\sep_{\cH}$, it suffices to 
bound the size of $\sep_{\cH}$ to prove the size property. From Lemma \ref{lem: weighted seperator} 
it follows that~$|\sep_{\cH}|=O(\sqrt{|\actF|})$. 
The next lemma bounds~$|\actF|$. 
%------------------------------------------------------------------------------------------
\begin{lemma}
$|\actF|=O(n\cdot 4^c)$, where $c$ is the number of color classes in~$V$.
\end{lemma}
%------------------------------------------------------------------------------------------
\begin{proof}
We first bound $|\actF_{i}|$, the number of active fragments created for
the segments in~$V_{i}$, in terms of the number of active fragments created
for $V_{<i}$.
\begin{claiminproof}
For $1< i \leq c$ we have  $|\actF_i| \leq 2\cdot |V_i| + 3\cdot (|\actF_{<i}| + 2n +1) -6$.
\end{claiminproof}
\vspace*{-5mm}
\begin{proofinproof}
The set $\actF_{i}$ contains at most $2 \cdot |V_{i}|$ end fragments. 
The internal fragments in $\actF_i$ connect two fragments 
from $\actF_{<i}$. We denote the set of these internal fragments by~$\actF_i^{\,\mathrm{int}}$.
Now consider the multi-graph~$\graph$ defined as follows.
\begin{itemize}
\item For each fragment in $\actF_{<i}$, we add a node to $\graph$.
\item For each fragment in $\actF_i^{\,\mathrm{int}}$, we add an edge
      to $\graph$ between the fragments from $\actF_{<i}$ that it connects.
\item For each endpoint of a segment in~$V_{\geq i}$ we add
      a singleton node to $\graph$. (Observe that the endpoints
      of segments in $V_{<i}$ lie on an end fragment in~$\actF_{<i}$, 
      which is already a node in $\graph$.) 
\end{itemize}
Now consider the obvious drawing of $\graph$, where the nodes
are drawn as fragments of~$\actF_{<i}$ or as points, and the edges are
drawn as fragment in $\actF_i^{\,\mathrm{int}}$.
Recall that two active fragments can touch, but they never cross.
By continuously shrinking the fragments in~$\actF_{\leq i}$ and deforming
the fragments of $\actF_i^{\,\mathrm{int}}$ appropriately, we can therefore 
create a plane drawing of $\graph$. That is, we can create a drawing 
of~$\graph$ in which the nodes are points
and the edges are pairwise disjoint curves connecting their endpoints.
See Figure~\ref{fig:fragmentstoplane} for an example of this deformation. For reasons that will become clear shortly,
we augment~$\graph$ with one additional singleton node $u_{\infty}$, which
we place in the unbounded face of~$\graph$.

\hspace*{1.5em}The graph $\graph$ is a multi-graph because $\actF_i^{\,\mathrm{int}}$
can contain multiple fragments connecting the same pair of fragments 
$f,f' \in \actF_{< i}$. Let $g,g'\in \actF_i^{\,\mathrm{int}}$ be 
two such fragments. The reason that we added both $g$ and $g'$ to
$\actF_i^{\,\mathrm{int}}$ is that $g$ and $g'$ were not equivalent.
Hence, the region $Q(f,f',g,g')$ % bounded by $f,f',g,g'$
contains an endpoint belonging to some segment~$v\in V$.
The deformation process that turns each node the drawing of $\graph$
into a point can be done in such a way that this property is maintained.
Thus, after the deformation we have a plane drawing of $\graph$ in which
for any two edges~$g,g'$ that connect the same pair of nodes, there
is a node inside the deformed region $Q(f,f',g,g')$.
Because of the additional node~$u_{\infty}$, we are also guaranteed
to have at least one node outside this region. A plane multi-graph
with this property is called a \emph{thin graph}. It is
known~\cite[Lemma 5]{10.1145/990308.990309} that the 
standard inequality 
$(\mbox{\#\,edges})\leq 3\cdot (\mbox{\#\,nodes})-6$  that holds for
planar graphs (with at least three vertices) also holds for thin graphs.
Hence, $|\actF_i^{\,\mathrm{int}}| \leq 3 (|\actF_{< i}|+2n+1)-6$.
\end{proofinproof}
Note that $|\actF_1| = |V_1|$ and $|V_i|  \leq n$ for all~$i$. Hence, the claim above gives us the %followibg 
recurrence~$|\actF_{\leq i}| \leq 4 |\actF_{< i}| + 8n - 3$ with~$|\actF_{\leq 1}| \leq n$.
This gives
$|\actF_{\leq i}| \leq \left(\tfrac{n}{4}+\tfrac{8n-3}{12}\right)\cdot 4^i-\tfrac{8n-3}{3}$. 
Plugging in $i=c$ gives $|\actF|=|\actF_{\leq c}|=O(n\cdot4^c)$, which proves the lemma.
\end{proof}
%---------------------------------------------------------------------------------
\begin{figure}[t]
\begin{center}
\includegraphics{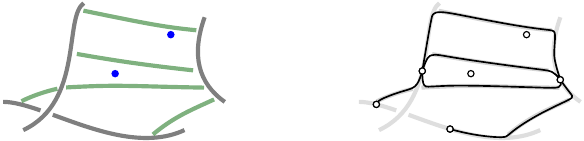}
\end{center}
\caption{Left: Fragments in $V_i$ are green, fragments in $\actF_i^{\,\mathrm{int}}$
         are dark gray, and endpoints of segments in $V_{\geq i}$ are blue. 
         Right: The plane multigraph created for the example on the left. }
\label{fig:fragmentstoplane}
\end{figure}
%---------------------------------------------------------------------------------
Since $c$, the number of color classes, is a constant, we obtain the following corollary.
%---------------------------------------------------------------------------------
\begin{corollary}
    The separator $\sep$ contains $O(\sqrt{n})$ star graphs.
\end{corollary}
%---------------------------------------------------------------------------------
So far we have considered input sets~$V$ where the segments are unweighted. 
To create a separator for $c$-oriented polygons, which we will do in the next 
section, we need a separator for weighted segments. Fortunately it is straightforward
to adapt the construction described above to the weighted setting---we only
need to change the weighting scheme we used in Step~2 of the construction.
More precisely, instead of assigning a weight $\tfrac{1}{n}$ to the representative
fragment of a segment $v$, we assign $\mathrm{weight}(v)/\sum_{u\in V}\mathrm{weight}(u)$ to the representative. 
\medskip

\noindent \emph{Remark.}
It is well known that grid graphs do not admit node-based separators of size $o(\sqrt{n})$. 
Because nodes in grid graphs have constant degree, bicliques in grid graphs have constant size.
Hence, grid graphs do not admit biclique-based separators of size $o(\sqrt{n})$. 
Grid graphs are bipartite and planar, which implies that they are {\sc pure-2-dir} graphs~\cite{HARTMAN199141}. 
We conclude that even {\sc pure-2-dir} graphs do not admit biclique based separators of size $o(\sqrt{n})$.

%---------------------------------------------------------------------------------
\subparagraph{Computation time.}
%---------------------------------------------------------------------------------
If we assume that the appropriate elementary operations on the pseudo-segments---computing
the intersection point of two pseudo-segments, for instance, or determining if a point
lies inside some region $Q(f,f',g,g')$---can be performed in $O(1)$ time, that then
a brute-force implementation of the algorithm presented above runs in polynomial time.
More interestingly, for $c$-oriented line segments, the algorithm 
can be implemented to run in $O(n\log n)$ time, as shown in Section \ref{sec:algorithms}.
We obtain the following theorem.
%---------------------------------------------------------------------------------
\begin{theorem}\label{thm: seperator for pseudo segments}
    Let $V$ be a $c$-colored set of $n$ non-negatively weighted pseudo-segments, 
    where~$c$ is a fixed constant, whose total weight 
    is at most~$1$. Then the intersection graph $\ig[V]$ has a star-based 
    separator~$\sep$ of size $O(\sqrt{n})$ such that $V\setminus\sep$ can be partitioned 
    into subsets $A,B$ of weight at most~$\tfrac{2}{3}$ with no edges between them. 
    The bound on the size of the separator is tight, even for axis-parallel segments.
    In the special case where $V$ is a set of $c$-oriented line segments,
    the separator $\sep$ and parts $A,B$ can be computed in $O(n\log n)$ time.
\end{theorem}
%---------------------------------------------------------------------------------

%---------------------------------------------------------------------------------
\subparagraph{Application to distance oracles.} \label{section:distance:oracles}
%---------------------------------------------------------------------------------
Arikati~\etal~\cite{DBLP:conf/esa/ArikatiCCDSZ96} presented a simple distance oracle for
planar graphs, using node-based separators. Aronov, De~Berg, and Theocharous~\cite{abt-cbsgd-24} 
observed that the approach can be adapted to work with clique-based separators, as follows. 
Let~$\graph=(V,E)$ be the graph for which we want to construct a distance oracle.
\begin{itemize}
\item Construct a clique-based separator~$\sep$ for~$\graph$, and let~$A,B\subseteq V\setminus S$
      be the two parts of the partition given by~$\sep$. For each node
      $v\in V$ and each clique $C\in \sep$, store the distance~$d(v,C) := \min \{ d(v,u) : u\in C\}$,
      where~$d(u,v)$ denotes the hop-distance from $s$ to $t$ in~$\graph$.
\item Recursively construct distance oracles for the subgraphs induced by
      the parts $A$ and~$B$.
\end{itemize}
Now suppose we want to answer a distance query with nodes~$s,t\in V$.
Let~$d^* := \min \{ d(s,C) + d(t,C): C\in \sep\}$. If $s$ and $t$ do not lie in the same
part---that is, we do not have~$s,t\in A$ or~$s,t\in B$---then we report $d^*$.
Otherwise, $s$ and $t$ lie in the same part of the partition, say~$A$.
Then we report the minimum of $d^*$ and the distance we obtain by querying the
recursively constructed oracle for~$A$.

This distance oracle uses $O(n\cdot s(n))$ storage, where $s(n)$ is the size of
the separator, and it has $O(s(n))$ query time, assuming $s(n)=\Omega(n^{\beta})$
for some constant $\beta>0$. The reported distance is either the exact distance
$d(s,t)$, or it is $d(s,t)-1$. The additive error of~$1$ is because we do
not know if $s$ and $t$ can reach the same node of some clique~$C$ with
paths of length $d(s,C) $ and $d(t,C)$, respectively---we may have to 
use an edge inside~$C$ to connect these paths.
We observe that the same approach can be used in combination with star-based
(or biclique-based) separators. The only difference is that we now get an additive 
error of at most~$2$, because we may need two additional 
edges inside a star (or biclique) in the separator. We obtain the following result.
%------------------------------------------------------------------------------------------
\begin{corollary}\label{cor:oracle}
Let $V$ be a $c$-colored set of pseudo-segments, where $c$ is a fixed constant.
There is an almost-exact distance oracle for $\ig[V]$ that uses $O(n\sqrt{n})$
storage and can report the hop-distance between any
two nodes $s,t\in V$, up to an additive error of~$2$, in $O(\sqrt{n})$ time.
\end{corollary}
%------------------------------------------------------------------------------------------

%------------------------------------------------------------------------------------------
\section{Extension to $c$-oriented polygons and string graphs}
\label{sec:polygons-and-strings}
%------------------------------------------------------------------------------------------

%------------------------------------------------------------------------------------------
\subparagraph{$c$-Oriented polygons.}
%\label{subsec:polygons}
%------------------------------------------------------------------------------------------
%------------------------------------------------------------------------------------------
Let $\cP=\{P_1,\ldots,P_n\}$ be a collection of $c$-oriented polygons, each 
with a constant number of edges, where $c$ is a fixed constant. The polygons may have holes.
We assume that the polygons in $\cP$ are in general position.
In particular, no three sides meet in a common point 
and no endpoint of one side lies on another side. The only exception is when two sides belong to the same polygon, in that case the two sides may share an endpoint. 
This assumption is without loss of generality, as it can always be ensured by perturbing the polygons slightly.

The idea is to create a weighted collection $V$ of segments 
to which we can apply Theorem~\ref{thm: seperator for pseudo segments},
and then use the resulting
separator~$\sep_V$ to construct a separator $\sep$ for~$\ig[\cP]$.
The set $V$ is created as follows.
\begin{itemize}
\item First, we add each side of every polygon $P_i \in \cP$ to~$V$. 
      For each polygon~$P_i$, we pick an arbitrary
      side as its \emph{representative side}, which we give 
      weight~$\tfrac{1}{n}$; other sides of $P_i$ are given weight~$0$. 
\item Second, to handle holes, we add \emph{connecting segments} to $V$. These segments will
      always have weight~$0$. All connecting segments will have orientation $\phi$, where $\phi$ be an orientation that is not used by any segment in $V$. We assume wlog that that $\phi$ is vertical.
      We handle each polygon separately. For each hole $H$ belonging to some polygon $P_i$ do the following: Take the topmost point $h_i$ of $H$, and let $\rho_i$ be a ray emanating from~$h_i$ with orientation~$\phi$. Let $h_i'$ be the point where $\rho_i$ intersects either another hole from $P_i$ or the outer boundary of $P_i$ for the first time. We add the connecting segment $h_i h_i'$ to $V$.
\item Finally, to handle the containment of polygons within other polygons, 
      we add so-called \emph{containment segments} to $V$. These segments will
      always have weight~$0$. For each polygon~$P_i$ let $\C_i \subseteq \cP$ be the set of 
      polygons that fully contain~$P_i$. Take a point~$x_i\in P_i$, and let $\C'_i \subseteq \cP$ be the set of polygons that contain $x_i$. Observe that $\C_i \subseteq \C'_i$. Let $\rho_i$ be a ray emanating from~$x_i$ with orientation~$\phi$.
      For each $P_j\in \C'_i$, let $y_j$ be the point where $\rho_i$ leaves $P_j$ for the first time. Let $j^*$ be such that $y_{j^*}$ is the last point among the points~$y_j$---in other words, the one furthest from $x_i$---and define $x'_i := y_{j^*}$.
      We now add $x_i x_i'$ as a containment segment to~$V$.
      Note that the containment segment $x_i x'_i$ 
      intersects the boundaries of all polygons in~$\C_i$, because it intersects the boundaries of all polygons in $\C_i'$. Moreover, $x_i x'_i$ is
      completely contained in $P_{j^*}$. Figure \ref{fig:containment} shows examples of containment segments. 

\end{itemize}
\begin{figure}[t]
\begin{center}
\includegraphics{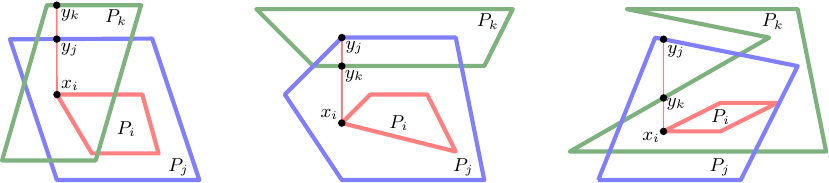}
\end{center}
\caption{
Left: The thin line is the containment segment for $P_i$. In this case $y_k=x_i'$. Note that while $P_k$ does not contain the $P_i$, it does contain $x_i$. Middle: The containment segment of $P_i$ can intersect $P_k$, even though $P_k$ does not intersect $P_i$. Right: The containment segment stops after intersecting the boundary of every polygon which contains $x_i$. Extending the containment segment could lead to further intersections with $P_k$. However, only the first intersection between the containment segment and a polygon is considered in the construction.}
\label{fig:containment}
\end{figure}

Because of the general position assumption, we can partition $V$ into color classes~$V_1,\ldots,V_{c+1}$ based on the 
orientation of the segments. We let $V_1$ contain the segments of orientation~$\phi$, 
which will be important later. The total weight of the segments is~$1$, 
so we can apply Theorem~\ref{thm: seperator for pseudo segments}. 
% specifically the separator for line segments. 
Let $\sep_V$ be the resulting
separator for $\ig[V]$, and let $A_V$ and $B_V$ be the two parts of weight
at most~$\frac{2}{3}$ into which $\sep_V$ splits $V\setminus \sep_V$.
We construct a separator $\sep$ and parts~$A,B$ for $\ig[\cP]$ as follows.

For each star in $\sep_V$ we consider its center~$v$. If $v$ is a side of
a polygon $P_i\in \cP$ or one of its connecting segments then we add $\mystar(P_i)$ to~$\sep$, where
$\mystar(P_i)$ is the subgraph of~$\ig[\cP]$ consisting of $P_i$ and its incident edges.
If $v$ is a containment segment~$x_i x'_i$
that we generated for polygon~$P_i\in \cP$, then let $P_{j^*}$ be the polygon that fully contains $x_ix'_i$ 
and where $x'_i\in \bd P_{j^*}$. We add $\mystar(P_{j^*})$ to~$\sep$.  
As before, we remove duplicate stars and we remove polygons from stars 
to ensure that each polygon is in at most one star.
To create the parts $A,B$, we consider the representative sides of the 
polygons $P_i\in\cP$ that are not in a star in~$\sep$. If the representative
side is in $A_V$, then we put $P_i$ in part~$A$; otherwise we put $P_i$ in part~$B$.

%------------------------------------------------------------------------------------------
\subparagraph{\rm\em The analysis.}
%------------------------------------------------------------------------------------------
Proving the size property, namely that $|\sep|=O(\sqrt{n})$, is easy. Indeed, 
per polygon we put $O(1)$ sides, $O(1)$ connecting segments, and at most one containment segment into $V$.
Hence, $|V| =O(n)$ and thus $|S|=|\sep_V| = O\left(\sqrt{n}\right)$. 
It remains to prove the separation property.
\medskip

Since the total weight of $A_V$ and of $B_V$ are both at most~$\tfrac{2}{3}$ and
each representative side has weight~$\tfrac{1}{n}$, it follows that
$|A|\leq \tfrac{2n}{3}$ and $|B|\leq \tfrac{2n}{3}$.  
Next, we prove that there are no edges from $A$ to $B$ in~$\ig[\cP]$.
We need the following lemma.
%------------------------------------------------------------------------------------------
\begin{lemma}\label{lem:consistent}
    Consider a polygon~$P_i \in\cP$. If $P_i\in A$ (resp.~$P_i \in B$), then all sides and connecting segments of $P_i$ are in $A_V$ (resp.~$B_V$). 
\end{lemma}
\begin{proof}
For ease of reading we consider the connecting segments of $P_i$ to be sides of $P_i$ within this proof.
First consider the case $P_i\in A$. Suppose for a contradiction that not all sides
of~$P_i$ are in $A_V$. Hence, $P_i$ has a side~$s\in\sep_V$ or a side $s'\in B_V$. 
We claim that in the latter case $P_i$ must also have a side $s\in\sep_V$. To see this,
consider the representative side $s_i$ of~$P_i$.
Note that $s_i\in A_V$ since we put $P_i$ into~$A$. 
Observe that the subgraph of $\ig[V]$ induced by the sides of $P_i$ is connected, 
because of the connecting segments that we added.
Thus there is a path in~$\ig[V]$ connecting $s'$ to $s_i$ and only using sides of $P_i$. Let us consider such a path. Since $s_i\in A_V$
and $s'\in B_V$, one of the nodes on this path must be in~$\sep_V$ and be a side of~$P_i$. This establishes our claim.

It remains to prove that if $P_i$ has a side~$s$ in $\sep_V$, then $P_i$ is a node in some star in $\sep$. If~$s$ is
the center of a star in $\sep_V$, then by construction $P_i$ is the center of a star
in $\sep$. If $s$ is a non-center node in some $\mystar(s')\in \sep_V$
we distinguish two cases.
If $s'$ is a side of some polygon~$P_j$, 
then~$P_j$ intersects~$P_i$ and~$\sep$ contains~$\mystar(P_j)$. Consequently, $P_i$ is a non-center node in $\mystar(P_j)$.
If $s'$ is a containment segment $x_j x_j'$ where $x_j'$ lies on some polygon $P_k$, then $P_k$ intersects $P_i$ because $x_j x_j'$ is fully contained in $P_k$. Because $\mystar(s')\in \sep_V$ it follows that $\mystar(P_k) \in \sep$, and thus $P_i$ is a node in some star in $\sep$.
We have reached a contradiction in all cases, it cannot be the case that $P_i\in A$ and that $P_i$ is a node in some star in $\sep$.  
We conclude that all sides of $P_i$ are in $A_V$.

Now consider the case $P_i\in B$. We can follow the proof for the case $P_i\in A$
if the representative side~$s_i$ of $P_i$ is in $B_V$. This must indeed
be the case: we cannot have $s_i\in A_V$ because then we would
have put $P_i$ into~$A$, and $s_i$ cannot be a node in a star
in~$\sep_V$ because then $P_i$ would have been in a star in~$\sep$, as has been shown above.
Hence, the lemma is also true if $P_i\in B$. 
\end{proof}

%---------------------------------------------------------------------------------
\begin{figure}[t]
\begin{center}
\includegraphics{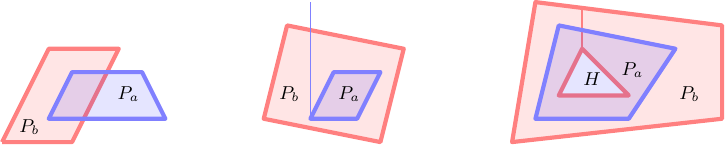}
\end{center}
\caption{Left: two intersecting polygons. Middle: $P_a$ is contained in $P_b$. The thin line represents the containment segment of $P_b$. The containment segments intersects $P_b$ but does not have to end at the boundary of $P_b$. Right: $P_a$ lies in the outer boundary of $P_b$. Polygon $P_b$ has a hole $H$ so~$P_a \not \subset P_b.$ The thin line represents the connecting segment of $H$.}
\label{fig:lemma12}
\end{figure}
%---------------------------------------------------------------------------------

%------------------------------------------------------------------------------------------
We can now prove that there are no edges from $A$ to $B$ in $\ig[\cP]$.
%------------------------------------------------------------------------------------------
\begin{lemma} \label{lemma:separate:polygons}
No polygon $P_a \in A$ intersects any polygon $P_b \in B$.
\end{lemma}
\begin{proof}
    Suppose for a contradiction that $P_a$ and $P_b$ intersect. We distinguish three cases, which are shown in Figure~\ref{fig:lemma12}.
    \begin{itemize}
    \item \emph{The boundaries of $P_a$ and $P_b$ intersect.}
           Let $s$ and $s'$ be sides of $P_a$ and $P_b$, respectively, that intersect. 
           Because $P_a\in A$ we have $s\in A_V$ by Lemma~\ref{lem:consistent}.
           Similarly, we have $s'\in B_V$.
           But this contradicts that $\sep_V$ is a separator for~$\ig[V]$ with parts~$A_V,B_V$.
    \item \emph{The boundaries of $P_a$ and $P_b$ do not intersect and one of the polygons is fully contained in the other.}
            Assume wlog that $P_a \subset P_b$. Consider the containment segment $s=x_a x'_a$.
            By construction of the containment segments, polygon~$P_b$
            has a side~$s'$ that intersects~$s$. 
            Recall that the containment segments were put into the set $V_1$ that
            was handled first by the algorithm from the previous section.
            Hence, there will be a fragment $f\in \actF$ that is identical to $s$.  
            Let $i>1$ be such that $s'\in V_i$. Then $f$ satisfies condition~(i)
            of Lemma~\ref{lem: assignment properties of blocks}, with $s'$
            playing the role of~$v$. Since $s'$ is a side of $P_b\in B$,
            we have $s'\in B_V$ by Lemma~\ref{lem:consistent}. Thus, 
            Lemma~\ref{lem: assignment properties of blocks} implies that $s=f\in B_V$.

            With a similar argument as above we will show that~$s=f \in A_V$. There exists a side~$s''$ of $P_a$ that intersects $s$, namely the side of $P_a$ that contains an endpoint of $s$. From Lemma~\ref{lem:consistent} it follows $s''  \in  A_V$. It must be that $s'' \in V_j$ for some $j >1$. Then $f$ also satisfies condition~(i) of Lemma~\ref{lem: assignment properties of blocks} with $s''$ playing the role of $v$, so that $s=f \in A_V$. But then the sets $A_V$ and $B_V$ are not disjoint, which contradicts that $\sep_V$ is a separator.  
    \item \emph{The boundaries of $P_a$ and $P_b$ do not intersect and none of the polygons is fully contained in the other.}        
            In this case, one of the polygons must be enclosed by the outer boundary of the other polygon. 
            Assume wlog that $P_a$ is enclosed by the outer boundary of~$P_b$.
            Since~$P_a \not\subset P_b$, there must the be a hole $H$ of $P_b$ that is contained in $P_a$. Because the subgraph of $\ig[V]$ induced by the sides and connecting segments of $P_b$ is connected, 
            % $H$ is contained within $P_a$, and $P_a$ is contained in the outer boundary of $P_b$, 
            there must thus be a side or connecting segment of $P_B$ that intersects the outer boundary of~$P_a$. Because the boundaries of $P_a$ and $P_b$ do not intersect, this must be some connecting segment $s_b$. Let $s_a$ be a side of $P_a$ intersected by~$s_b$. It follows from Lemma~\ref{lem:consistent} that~$s_a \in V_A$ and $s_b \in V_B$, which contradicts that $\sep_V$ is a separator.  \qedhere
    \end{itemize}
\end{proof}
\medskip
%------------------------------------------------------------------------------------------
Putting everything together, we obtain the following theorem. The runtime guarantee for 
computing the separator is proven in Section \ref{sec:algorithms}.
%------------------------------------------------------------------------------------------
\begin{theorem} \label{thm:polygons}
Let $P$ be a set of $n$ constant-complexity $c$-oriented polygons in the plane. 
Then the intersection graph~$\ig[P]$ has a star-based separator of size $O(\sqrt{n})$,
which can be computed in $O(n \log n)$ time. Moreover, there is an almost exact distance oracle for
$\ig[V]$ that uses $O(n\sqrt{n})$ storage and that can report the hop-distance between any
two nodes $s,t\in P$, up to an additive error of~$2$, in $O(\sqrt{n})$ time.
\end{theorem}
%------------------------------------------------------------------------------------------

%------------------------------------------------------------------------------------------
\noindent\emph{Remark.}
%------------------------------------------------------------------------------------------
The theorem above is stated for $n$ constant-complexity $c$-oriented polygons. However, it
also holds in a more geneal setting, namely for a collection $\cP$ of polygons with $n$ edges in total.
We can then find a separator of size $O(\sqrt{n})$ and parts $A,B$ such that the number
of polygons in $A$ and $B$ is at most $\tfrac{2|\cP|}{3}$. Alternatively, we can guarantee
that the total number of edges of the polygons in~$A$ (and similarly for $B$) is at most~$\tfrac{2n}{3}$.
Finally, the theorem also works in a weighted setting.

%------------------------------------------------------------------------------------------
\subparagraph{String graphs.}
%------------------------------------------------------------------------------------------
A \emph{string graph} is the intersection graph of a set $V$ of curves in the 
plane~\cite{DBLP:journals/jct/EhrlichET76}---no conditions are put on the curves and,
in particular, any two curves in~$V$ can intersect 
arbitrarily many times. (But note that there is still at most one edge between 
the corresponding nodes in~$\ig[V]$.) It is known that for any set $U$ of connected
regions in the plane, there is a set $V$ of strings such that $\ig[U]$ and $\ig[V]$ are 
isomorphic~\cite{Lee-string-sep}. Thus, string graphs are the most general type of intersection graphs of 
connected regions in the plane. 
\medskip

Matou\v{s}ek~\cite{Matousek14}
proved that any string graph with $n$ nodes and $m$ edges has a (node-based) separator
of size~$O(\sqrt{m} \log m)$.\footnote{A paper by Lee~\cite{Lee-string-sep} claims that a separator of size $O(\sqrt{m})$ exists for string graphs, but Bonnet~\etal\cite{bonnet2024treewidthpolynomialmaximumdegree} note that there is an error in this paper. It is not yet known if the proof can be repaired.} Using this result we can obtain a star-based separator of sublinear
size for a string graph $\ig[V]$ using the following simple two-stage process. 
\begin{itemize}
\item \emph{Stage 1:} As long as there is a node $v\in V$ of degree at least $n^{1/3}/\log^{2/3} n$,
       remove $\mystar(v)$ from $\ig[V]$ and add it to the separator. 
       This puts at most $O(n^{2/3} \log^{2/3} n)$ stars into the separator. 
\item \emph{Stage 2:} Construct a node-based separator on the remaining string graph
      using Matou\v{s}eks method~\cite{Matousek14}, and put these nodes as singletons into the separator.
      Since the maximum degree after Stage~1 is~$O(n^{1/3}/\log ^{2/3} n)$, the remaining graph has
      $O(n^{4/3}/\log^{2/3} n)$ edges. Hence, in Stage~2 we put $O\left(\sqrt{\frac {n^{4/3}}{\log^{2/3}n }} \cdot  \log \left (\frac {n^{4/3}}{\log^{2/3}n }\right) \right)= O\left(n^{{2}/{3}} \log ^{{2}/{3} }n\right)$ stars
      into the separator.
\end{itemize}
By using the star-based separator exactly as in Section~\ref{section:distance:oracles}, this yields the following.
%------------------------------------------------------------------------------------------
\begin{proposition} \label{prop:strings}
Let $V$ be a set of $n$ strings in the plane. Then the string graph~$\ig[V]$ has a star-based
separator of size $O(n^{2/3} \log ^{2/3}n)$. Moreover, there is an almost exact distance oracle for
$\ig[V]$ that uses $O(n^{5/3} \log ^{2/3}n)$ storage and that can report the hop-distance between any
two nodes $s,t\in V$, up to an additive error of~$2$, in $O(n^{2/3} \log ^{2/3}n)$ time.
\end{proposition}
%------------------------------------------------------------------------------------------
Recall that the distance oracle of Aronov, De~Berg~and Theocharous~\cite{bkmt-cbsgis-23} 
for geodesic disks in the plane has $O(n^{7/4+\eps})$ storage and $O(n^{3/4+\eps})$ query time.
Thus the distance oracle in Proposition~\ref{prop:strings} is more general (as it can handle
any string graph), has better storage, and better query time. The only downside is that the
additive error in our distance oracle is at most~2, while for their oracle it is at most~1.
\medskip

\noindent \emph{Remark.}
One may wonder if any graph, and not just any string graph, admits a sublinear
star-based separator, but this is not the case.
For example, 3-regular expanders do not admit sublinear node-based separators~\cite{Expanders_no_seperator}
and working with star-based separators instead of node-based separators
does not help in constant-degree graphs.

\section{Efficient algorithms to construct the separators}
\label{sec:algorithms}
%-----------------------------------------------------------------------------------
In this section we provide algorithms that can compute star-based separators for $c$-oriented line segments and $c$-oriented polygons in $O(n \log n)$ time. We start by providing an algorithm for $c$-oriented line segments in general position. This algorithm will be used to create an algorithm for $c$-oriented polygons in general position. At the end of this section,
we explain how to handle degenerate cases.

%-----------------------------------------------------------------------------------
\subsection{$c$-Oriented line segments}
\label{subsec:c-oriented-segments}
%-----------------------------------------------------------------------------------
%-----------------------------------------------------------------------------------
When $V$ is a set of $c$-oriented line segments, where $|V|=n$, then we can compute 
a star-based separator in $O(n\log n)$ time, as described next. 
The algorithm will follow the construction used in 
Theorem~\ref{thm: seperator for pseudo segments}, except that we implement Step~1, 
where we construct the active fragments, slightly differently. 
For now we assume the segments are in general position. More precisely, we assume 
that no two segments of the same orientation overlap. Note that other degenerate cases 
can exist as well. For instance, an endpoint of one segment can be contained in
another segment, or three or more segments may pass through the same point.
While we do not address these types of degeneracies explicitly in this section, 
the algorithm can easily be adapted to deal with them. 
How we handle overlapping segments 
will be explained in Section~\ref{subsec:algo-degnerate}.

%-----------------------------------------------------------------------------------
\subparagraph{Implementation of Step~1.}
%-----------------------------------------------------------------------------------
Let $V_1,\ldots,V_c$ be a partition of the segments in~$V$,  based on their orientation. 
Since the segments are in general position, no pair of segments from the same set~$V_i$ intersect.
As before, we construct the set $\actF=\actF_1\cup\cdots\cup\actF_c$ 
of active fragments by handling the subsets~$V_1,\ldots,V_c$ one by one.
We set $\actF_1 := V_1$, and create the sets $\actF_i$ for~$i>1$ as follows.

Assume wlog that the segments in $V_i$ are horizontal.  
Let $L_i$ be the set containing the fragments in $F_{<i}$ and the endpoints of 
segments in $V_{\geq i}$, where we consider the endpoints to be fragments of length~0. 
Note that $|L_i| \leq |F_{<i}| + 2n$.
We start by computing the \emph{horizontal decomposition} of $L_i$.
This is the subdivision of the plane obtained by drawing, for each endpoint of a fragment $f\in L_i$,
a \emph{horizontal extension} that extends from $f$ to the right and to the left
until a fragment in~$\actF_i$ or, if no fragment is hit, to infinity.
%; see Figure~\ref{fig:h-decomp}.
This decomposition, which we denote by $\hd_i$, can be computed in $O(|L_i|\log |L_i|)$ time
with a simple plane-sweep algorithm~\cite{bcko-cgaa-08}.
The horizontal decomposition~$\hd_i$ consists of at most $3|L_i|+1$ (possibly
unbounded) trapezoids~\cite[Lemma 6.2]{bcko-cgaa-08},
each of which is bounded by one or two horizontal extensions and by 
at most two fragments from~$\actF_i$.
We now compute the set $\actF_i$ of active fragments as follows.
\medskip

First, we compute all end fragments and add them to $\actF_i$. 
There are at most $|V_i|$ of these segments.
Since each endpoint of a segment $v\in V_i$ was used in the creation of the horizontal 
decomposition, the endpoint has horizontal extensions to the left end right. 
By selecting the neighbor in the direction of the other endpoint of~$v$, 
this allows us to determine the corresponding end fragment in constant time.

Next, we compute the active internal fragments.
Due to our assumption that the segments in $V_i$ are horizontal, each internal fragment that
should be created from a segment in $V_i$ consists of a horizontal segment that fully crosses 
some \emph{bounded} trapezoid $\Delta$ of $\hd_i$ from its left to its right edge. 
Note that the left and right edges of $\Delta$ are parts of fragments in~$\actF_{<i}$. 
To compute a suitable set of active internal fragments, we will therefore go through all bounded trapezoids using a plane sweep. For each such bounded trapezoid $\Delta$, we try to find a segment $v$ that fully crosses it.
If such a segment exists, then we add $v\cap \Delta$ as an active fragment to $\actF_i$.
Note that there can also be other segments $u\in V_i$ that fully cross $\Delta$,
but the fragments $u\cap \Delta$ are all equivalent to $v\cap\Delta$, 
so we do not need to make them active. 
We make at least one fragment
from each equivalence class active, which is sufficient for the correctness of our algorithm. We can make more than one fragment from an equivalence class active with this construction, however, this does not affect the correctness of the algorithm.  

To find the active internal fragments efficiently, we use a plane-sweep algorithm,
using a horizontal sweep line~$\ell$ that sweeps from top to bottom over the plane.
During the sweep, we maintain a set of trapezoids intersecting~$\ell$,
stored in left-to-right order in a balanced binary search tree~$\T$. When we 
encounter the top edge of a trapezoid, we insert it into~$\T$. 
When we encounter a segment~$v\in V_i$, we report all trapezoids that
are fully crossed by~$v$; using the tree~$\T$ this can be done in $O(\log |L_i|+k_v)$ time, 
where $k_v$ is the number of such trapezoids. For each reported trapezoid~$\Delta$,
we add $v\cap \Delta$ to $\actF_i$, after which we remove $\Delta$ from~$\T$. 
When we encounter a bottom edge of a trapezoid~$\Delta$, then we also delete $\Delta$
from~$\T$ (if it is still present).
Because we remove trapezoids which are found in a query, the sum over~$k_v$ for all $v\in V_i$ is bounded by $|L_i|$.
This way all active internal fragments can be computed in~$O(|L_i|\log |L_i|)$ time in total.

Since each segment in $V_i$ contributes at most two end fragments, while each of the at most $3|L_i|+1$ bounded trapezoids contributes at most one active fragment, we have 
\[
|\actF_{\le i}|\le|\actF_{<i}|+3|L_i|+1+2|V_i| \le 4|F_{<i}|+8n+1.
\]
We can now compute the total number of active fragments created from the following recurrence:
\[
    |\actF_{\leq i}| \leq \begin{cases}
    n  & \mbox{ if $i = 1$} \\
    4|\actF_{< i}| +8n +1   & \mbox{ if $1< i \leq c$} \end{cases}
    \]   
Solving this recurrence we get $|\actF_{\leq i}| = O(n\cdot4^i)$. 
Plugging in $i=c$ gives 
\[ 
|\actF|=|\actF_{\leq c}|=O(n\cdot4^c)=O(n).
\]

Because $|L_i|\le|F_{<i}|+2n$ and we only have a constant number of iterations 
of the algorithm that computes active fragments, we can compute $\actF$ in time $O(n \log n)$ time.

%-----------------------------------------------------------------------------------
\subparagraph{Implementation of Step~2.}
%-----------------------------------------------------------------------------------
To construct the contact fragment graph $\cH$ on the set of fragments $F$ we need to find the intersection points
of fragments in $F$. Because the fragments do not cross but only touch, each intersection point
coincides with an endpoint of a fragment. Hence, we can find the $O(n)$ intersection points 
and create the graph $\cH$ in $O(n \log n)$ time using a standard plane-sweep algorithm~\cite{bcko-cgaa-08}. 
Finding a weighted separator $\sep_{\cH}$ for $\cH$ 
can be done in $O(n)$ time~\cite{planarseparator}.

%-----------------------------------------------------------------------------------
\subparagraph{Implementation of Step~3.}
%-----------------------------------------------------------------------------------
In Step~3 we convert the separator $\sep_{\cH}$ for $\cH$ into the  separator~$\sep$ for~$\ig[V]$.
To this end, we first determine in $O(|\sep_{\cH}|)$ time the set $V^*\subseteq V$ of segments
that form the centers of the stars in~$\sep$. 
We then compute the stars of the segments in~$V^*$,
making sure that no segment ends up in more than one star.
Let $V^*_i := V^*\cap V_i$. We go over the sets~$V^*_i$ one by one, and
will compute the stars for the segments in $V^*_i$ as described next. 
\medskip

Consider a set $V_i^*$ and assume wlog that the segments in $V^*_i$ are horizontal.
To create the stars of the segments $v\in V^*_i$, we go over the sets $V_j$ for $j\neq i$ 
one by one. Here we ignore segments $u\in V_j$ that have already been added to
the star of a segment in $V^*_{<i}$.
To handle $V_j$, we move a horizontal sweep line from top to bottom over the plane,
maintaining (a subset of) the segments from $V_j$ that intersect the sweep line.
These segments are stored in a balanced binary tree~$\T$, ordered from left to right. 
Note that the order of the segments stored in $\T$ does not change, since all are parallel. 
When we encounter the top endpoint of a segment $u\in V_j$ we add $u$ to~$\T$. When we
encounter a segment~$v\in V^*_i$, we search in~$\T$ to find all segments in~$\T$
that intersect~$v$, we add them to $\mystar(v)$ and we delete them from~$\T$.
When we encounter the bottom endpoint of a segment $u\in V_j$ we delete $u$ from~$\T$
(if it is still present). The plane sweep takes $O((|V^*_i|+|V_j|)\log (|V^*_i|+|V_j|))$
time in total.

After handling each of the sets $V_j$, we have computed the stars of the segments in~$V^*_i$.
The total time for this is $O(\sum_{j\neq i} ((|V^*_i|+|V_j|)\log (|V^*_i|+|V_j|)) )= O(n\log n)$.
Since we have~$c=O(1)$ sets $V^*_i$, the total time to compute $\sep$ is $O(n\log n)$.
Note that the parts~$A,B\subseteq V$ of the partition induced by $\sep$
can be determined in $O(n)$ time from the partition of $\cH$ given by $\sep_{\cH}$.
\medskip

We conclude that for a $c$-oriented set of $n$ line segments in general position,
a separator with the properties of Theorem~\ref{thm: seperator for pseudo segments}
can be computed in~$O(n\log n)$ time.

%-----------------------------------------------------------------------------------
\subsection{$c$-Oriented polygons}
%-----------------------------------------------------------------------------------
We now prove the runtime claim from Theorem \ref{thm:polygons}. 
Let $\cP=\{P_1,\ldots,P_n\}$ be a collection of constant-complexity $c$-oriented polygons, 
possibly with holes, where $c$ is a fixed constant. We assume
wlog that none of the $c$ orientations of the polygon edges is vertical. 
We also assume that the polygons are in general position, leaving the treatment
of degenerate cases to Section~\ref{subsec:algo-degnerate}.
%-----------------------------------------------------------------------------

%-----------------------------------------------------------------------------
Recall that the algorithm described in Section~\ref{sec:polygons-and-strings} 
computes a star-based separator $S$ for~$\ig[\cP]$ as follows.
\begin{enumerate}
\item \label{step:pol1} 
      Compute the connecting and containment segments for the polygons $P_i\in\cP$.
\item \label{step:pol2}
      Compute a star-based separator $S_V$ for the (weighted) set $V$ of
      polygon sides and connecting and containment segments.
\item \label{step:pol3}
      Based on the separator $S_V$, determine a suitable collection $\cP^*$ of $O(\sqrt{n})$
      polygons $P_i$ and put the corresponding stars $\mystar(P_i)$ into~$S$,
      making sure that the stars are disjoint.
\end{enumerate}
Above we already described how to implement Step~\ref{step:pol2} in $O(n\log n)$
time,
so it remains to discuss Steps~\ref{step:pol1} and~\ref{step:pol3}. Note that for step 2 it is required that no 2 segments overlap. We will deal with overlapping sides of polygons in Section~\ref{subsec:algo-degnerate}. In this section we assume no sides of polygons overlap and that all endpoints of sides have unique $x$-coordinates.  

\subparagraph{Computing connecting and containment segments.}
Consider a polygon $P_i\in \cP$. Recall that the connecting segments of $P_i$
are vertical segments that connect a topmost vertex $x_H$ of each hole $H$ of $P_i$ to the 
edge of~$P_i$ immediately above~$x$. Thus, the connecting segments of $P_i$
can be computed in $O(|P_i|\log |P_i|)$ time with straightforward plane-sweep algorithm.
The total time to compute all connecting segments is therefore $O(n\log n)$.
\medskip

To construct the containment segments, we proceed as follows. Pick a point 
$x_i$ in each polygon~$P_i$ and define $X := \{x_i : P_i \in \cP\}$. 
The points $x_i$ can be chosen arbitrarily, as long as they have distinct $x$-coordinates
(so that we will not generate overlapping containment segments). As before, 
let $\C'_i$ be the set of polygons containing the point~$x_i$, 
and let $\rho_i$ be the ray starting at $x_i$ and going vertically upward. Recall that for $P_j\in \C'_i$, 
the point $y_j$ is the first point where $\rho_i$ exits~$P_j$, and that the containment segment for $P_i$ 
is $x_i x'_i$, where $x'_i$ is the highest point in the set $\{y_j : P_j \in \C'_i\}$. 
Next we describe how to compute the points $x'_i$. 

First, we compute the vertical decomposition $\vd(P_i)$ of each $P_i$, thus decomposing $P_i$ into trapezoids 
with two vertical sides (one of which can be of zero length). Let $\Gamma$ be the set of trapezoids 
created over all polygons $P_i\in\cP$. We partition $\Gamma$ into $c^2$ classes 
$\Gamma_1,\ldots, \Gamma_{c^2}$ according to the orientations of their bottom and top sides. 
In other words, two trapezoids are in the same class  iff their bottom sides are parallel to each other 
and their top sides are parallel to each other. 
For a point $x_i\in X$ and a  class $\Gamma_k$, let $\Gamma_k(x_i)$ 
be the set of trapezoids in $\Gamma_k$ that contain $x_i$ in their interior. 
To compute the points $x'_i$ it suffices to determine, for each class $\Gamma_k$ and each $x_i\in X$, 
the highest point where $\rho_i$ exits a trapezoid in $\Gamma_k(x_i)$; the point $x'_i$ is 
then the highest such point among all $c^2$ classes. 
\medskip

To handle a class $\Gamma_k$, we use a plane-sweep algorithm that sweeps a vertical line~$\ell$ 
from left to right over the plane. During the sweep, we maintain the set $\Gamma_k(\ell)$
of trapezoids $\Delta\in\Gamma_k$ that intersect~$\ell$. The set  $\Gamma_k(\ell)$ is stored as follows.

Consider a coordinate system with one axis that is orthogonal to the bottom sides of the 
trapezoids in $\Gamma_k$ and one axis that is orthogonal to the top sides.
We call the former axis the $\alpha$-axis and the latter axis the $\beta$-axis,
and we direct both axes upward. 
Let $\alpha(\Delta)$ be the $\alpha$-coordinate of the bottom side of $\Delta$
and let $\beta(\Delta)$ be the $\beta$-coordinate of its top side; see  Figure~\ref{fig:alpha-beta}(i).
%-----------------------------------------------------------------------------------
\begin{figure}
\begin{center}
\includegraphics{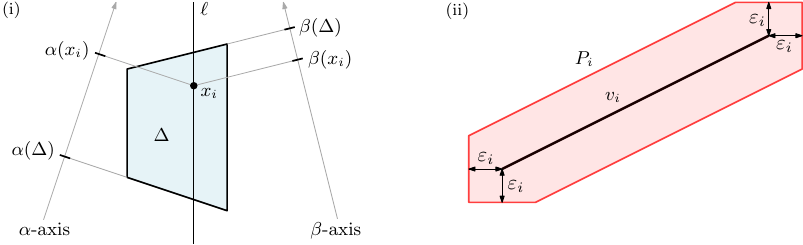}
\end{center}
\caption{(i) The $\alpha$- and $\beta$-coordinates of a trapezoid $\Delta$ and point~$x_i$. (ii) A line segment $v_i$ and its associated polygon $P_i$.}
\label{fig:alpha-beta}
\end{figure}
%-----------------------------------------------------------------------------------
We maintain the set $\{(\alpha(\Delta),\beta(\Delta)) : \Delta\in\Gamma_k(\ell)\}$ 
in a priority search tree~~$\T$~\cite{bcko-cgaa-08}.
The priority search tree $\T$ allows us to report, for a given query value $\alpha^*$,
in~$O(\log n)$ time the trapezoid~$\Delta\in \Gamma_k(\ell)$ with maximum $\beta$-coordinate
among the trapezoids with $\alpha(\Delta)\leq \alpha^*$. It also allows us to
report all trapezoids with $\alpha(\Delta)\leq \alpha^*$ and $\beta(\Delta)\geq \beta^*$,
for query values~$\alpha^*,\beta^*$. This takes $O(\log n + \mbox{number of reported trapezoids})$ time.
When the sweep line $\ell$ reaches the left side of a trapezoid~$\Delta$
we insert $(\alpha(\Delta),\beta(\Delta))$ into $\T$ and when $\ell$ reaches the 
its right side we delete $(\alpha(\Delta),\beta(\Delta))$ from~$\T$.

The sweep line also halts when we encounter a point $x_i\in X$. Query $T$ with $\alpha(x_i)$ to get trapezoid $\Delta^*$
with maximum $\beta$-coordinate among all 
$\Delta\in\Gamma_k(\ell)$ with $\alpha(\Delta)\leq \alpha(x_i)$, 
where $\alpha(x_i)$ is the $\alpha$-coordinate of~$x_i\in X$.
If $\beta(\Delta^*)\geq \beta(x_i)$ then $\Delta^*$ determines the highest exit point
among the trapezoids in $\Gamma_k$ that contain~$x_i$; otherwise $x_i$ is not contained
in any trapezoid in $\Gamma_k$.

Since insertions, deletions, and queries in priority search tree all take $O(\log n)$ time,
this plane-sweep algorithm runs in $O(n_k\log n_k)$ time, where $n_k := |\Gamma_k|+|X|$.
Thus the total time to handle all sets $\Gamma_k$ is $\sum_{k=1}^{c^2} O(n_k\log n_k) = O(n\log n)$.

\subparagraph{Computing the stars of the separator.}
Recall that Step~\ref{step:pol2} of our algorithm provides us with a collection $\cP^*$
of polygons that will form the centers of the stars in the separator~$S$.
To compute the stars themselves, we need to compute for each $P_j\in\cP\setminus \cP^*$ 
one center $P_i\in\cP^*$ that intersects $P_j$, if such a center exists, and then
add $P_j$ to $\mystar(P_i)$. We do this as follows.

We start by computing the vertical decomposition $\vd(P_i)$ of each polygon $P_i\in \cP$.
A polygon $P_j$ intersects a polygon $P_i$ iff a trapezoid $\Delta\in \vd(P_j)$ intersects
a trapezoid~$\Delta'\in\vd(P_j)$, which happens when at least one one of the following 
conditions is satisfied: an edge of $\Delta$ intersects an edge of $\Delta'$,
or a vertex of $\Delta$ is contained in $\Delta'$, or $\Delta$ contains a vertex of~$\Delta'$.
We treat each case separately. In each of the cases, as soon as we discover that
a polygon $P_j\in \cP\setminus \cP^*$ intersects a polygon $P_i\in\cP^*$, we 
ignore the trapezoids of $\vd(P_j)$ from that moment on. (This may involve deleting 
these trapezoids from the relevant data structures, but this will not increase 
the overall time bound.)
\begin{itemize}
\item Determining for each trapezoid~$\Delta$ of some polygon $P_i\in\cP^* \setminus \cP^*$ if one
      of its edges intersects an edge of a trapezoid~$\Delta'$ of some $P_j\in\cP^*$,
      and reporting the corresponding $P_j$, can be done in the same way as we constructed the stars in the case of  $c$-oriented segments; see the end of Section~\ref{subsec:c-oriented-segments}. 
\item Determining for each trapezoid~$\Delta$ of some polygon $P_i\in\cP^* \setminus \cP^*$ if one
      of its vertices is contained in a trapezoid~$\Delta'$ of some $P_j\in\cP^*$,
      and reporting the corresponding $P_j$, can be done with a plane-sweep
      algorithm that is similar to the one used to compute the containment segments:
      we maintain the trapezoids belonging to polygons in $\cP^*$ intersecting the sweep line, and when we reach a vertex~$x$
      of a trapezoid~$\Delta$ of some polygon in $\cP^* \setminus \cP^*$, we use the priority search tree
      to determine a trapezoid belonging to some $P_j\in\cP^*$ containing~$x$, if it exists.
\item  Determining for each trapezoid~$\Delta$ of some polygon $P_j\in\cP^* \setminus \cP^*$ if
      it contains a vertex $x$ of a trapezoid~$\Delta'$ of some $P_j\in\cP^* $,
      can be done with a plane-sweep algorithm as well. This time, when we encounter a vertex~$x$,
      we report all trapezoids $\Delta$ that contain~$x$. (As already mentioned,
      these trapezoids now need to be deleted form the priority search tree.)
\end{itemize}
The total running time of our algorithm to compute the stars is easily seen to be~$O(n\log n)$.
This finishes the proof that the separator described in Theorem~\ref{thm:polygons} 
can indeed be computed in $O(n \log n)$ time, assuming general position assumptions.

%--------------------------------------------------------------------------
\subsection{Handling degenerate cases}
\label{subsec:algo-degnerate}
%--------------------------------------------------------------------------
In this section, we show that computing a separator for $c$-oriented line segments 
can even be done in $O(n \log n)$ time when the segments are not in general position. 
Specifically, we deal with the case of overlapping segments, thus extending the algorithm
from {\sc pure-$c$-dir} graphs (with an explicit representation) to 
{\sc $c$-dir} graphs (with an explicit representation).
This immediately implies that the algorithm for $c$-oriented polygon also works on
degenerate cases, because the general-position assumption was only needed to ensure 
that the set $V$ of polygon sides, connecting segments, and containment segment 
was non-overlapping.

Let $V$ be a set of $c$-oriented line segments, where $|V|=n$. We inflate 
the line segments into thin polygons. Specifically, for each $v_i \in V$ we choose some small 
value~$\varepsilon_i$ and create a polygon $P_i$ containing all the points whose 
$\ell_{\infty}$-distance to $v_i$ is at most $\varepsilon_i$, see Figure~\ref{fig:alpha-beta}(ii). We thus create an at most 
$(c+2)$-oriented set of polygons~$\cP$. We can choose the $\varepsilon_i$ in such a way that the generated set of polygons is in general position and any two polygons $P_i,P_j$ intersect
iff the segments $v_i,v_j$ intersect. This means that $\ig[V]$ is equivalent to $\ig[\cP]$.

To create the polygons, we first calculate the smallest $\ell_{\infty}$-distance,
$d_{\min}$, between two non-intersecting segments. We then choose the $\varepsilon_i$ 
in such a way that that every $\varepsilon_i$ is smaller than $d_{\min}/2$ 
and all values are unique.
The distance $d_{\min}$ can be found by computing, for each segment endpoint~$p$, 
the distance to the segment closest to $p$ for which $p$ is not an endpoint. 
Because the segments can be partitioned into $c$ sets where each sets contains 
only parallel segments, this can be done a plane-sweep algorithm on each of the sets
and then returning the minimum over all sets. The plane-sweep algorithms 
take $O(n \log n)$ time, and since $c$ is constant, the total runtime is $O(n \log n)$.  
(Instead of computing explicit values for the $\varepsilon_i$'s, we can also
do the computations symbolically, as in the Simulation-of-Simplicity 
framework~\cite{simulation_of_symplicity}.) 
%Thus, $\ig[V]$ is equivalent to $\ig[\cP]$ and we can create the set $\cP$ in $O(n \log n)$ time. 
After creating the set $\cP$ of inflated polygons, we can use the $O(n \log n)$
algorithm for finding a separator for $c$-oriented polygons in general position. 
Because $\ig[V]$ is equivalent to $\ig[\cP]$, this gives us a separator 
for $\ig[V]$.

%------------------------------------------------------------------------------------------
\section{Concluding remarks}
\label{sec:concl}
%------------------------------------------------------------------------------------------
Motivated by the fact that intersection graphs of non-fat objects may not admit sublinear node-based
or clique-based separators, we introduced a \emph{biclique-based} and \emph{star-based} separators. 
We proved that the intersection graph of any $c$-colored set of pseudo-segments has a star-based separator
of size $O(\sqrt{n})$, and extended the result to $c$-oriented polygons. We also presented
a straightforward algorithm to compute a star-based separator of size~$O(n^{2/3} \log ^{2/3} n)$ for any
string graph. These results lead to almost exact distance oracles with subquadratic storage
and sublinear query times. To the best of our knowledge, such distance oracles did not yet exist---not 
even for intersection graphs of axis-parallel line segments.

Our work raises several questions. Can we improve the size of star-based separators
for string graphs from $O\left(n^{{2}/{3}} \log ^{{2}/{3} }n\right)$ to $O(\sqrt{n})$? If not, can we perhaps do so for
$c$-colored sets of strings, or for arbitrary sets of line segments? It is also
interesting to explore other applications of biclique-based separators, besides distance
oracles, and to see if the bounds we obtained for distance oracles can be improved. 
While clique-based separators have been used to design subexponential algorithms for problems 
such as \qcol~\cite{bkmt-cbsgis-23} and \domset~\cite{bbkmz-ethf-20}, it is unlikely 
that our biclique-based separator will yield new results for these problems. 
This is due to existing $2^{\Omega(n)}$ conditional lower bound (under ETH)
for \qcol and \domset on $2$-{\sc dir} and segment intersection graphs, respectively~\cite{BonnetR19}. 
Instead, problems whose main difficulty lies in finding (hop-)distances---computing the diameter in subquadratic time~\cite{DBLP:conf/compgeom/Chang0024} 
is an example---would be interesting to consider.

Recent work on the (weighted) \mis problem on restricted graph classes has exploited properties that can be 
interpreted through the lens of star-based separators~\cite[\S 1.4]{Gartland23}. It is known that for every constant~$t$, the family of 
graphs that does not contain an induced path of length~$t$ admits balanced separators that consist of the neighborhoods of~$t-1$ vertices (cf.~\cite[Thm. 1.2]{GartlandLMPPR24}). 
Hence such graphs have star-based separators of constant size. This property has been used to develop quasi-polynomial-time 
approximation schemes for weighted \mis on~$P_t$-free graphs~\cite{ChudnovskyPPT20,GartlandLMPPR24}, as well as exact subexponential-time algorithms for unweighted \mis on~$P_t$-free graphs~\cite{BacsoLMPTL19}. Can our star-based separators of size $O(\sqrt{n})$ also be used to obtain new algorithms for restricted input families?

%------------------------------------------------------------------------------------------

%------------------------------------------------------------------------------------------
\bibliography{references,literature}
%------------------------------------------------------------------------------------------

\end{document}